\documentclass[12pt,reqno]{article}

\usepackage[usenames]{color}
\usepackage{amssymb}
\usepackage{amsmath}
\usepackage{amsthm}
\usepackage{amsfonts}
\usepackage{amscd}
\usepackage{graphicx}

\usepackage[colorlinks=true,
linkcolor=webgreen,
filecolor=webbrown,
citecolor=webgreen]{hyperref}

\definecolor{webgreen}{rgb}{0,.5,0}
\definecolor{webbrown}{rgb}{.6,0,0}

\usepackage{color}
\usepackage{fullpage}
\usepackage{float}

\usepackage{graphics}
\usepackage{latexsym}
\usepackage{epsf}
\usepackage{breakurl}

\newcommand{\seqnum}[1]{\href{https://oeis.org/#1}{\underline{#1}}}

\def\modd#1 #2{#1\ \mbox{\rm (mod}\ #2\mbox{\rm )}}
\DeclareMathOperator{\PalFac}{PalFac}
\DeclareMathOperator{\Fac}{Fac}
\DeclareMathOperator{\Lead}{Lead}

\begin{document}

\theoremstyle{plain}
\newtheorem{theorem}{Theorem}
\newtheorem{corollary}[theorem]{Corollary}
\newtheorem{lemma}[theorem]{Lemma}
\newtheorem{proposition}[theorem]{Proposition}

\theoremstyle{definition}
\newtheorem{definition}[theorem]{Definition}
\newtheorem{example}[theorem]{Example}
\newtheorem{conjecture}[theorem]{Conjecture}

\theoremstyle{remark}
\newtheorem{remark}[theorem]{Remark}

\title{Words With Few Palindromes, Revisited}

\author{Lukas Fleischer and Jeffrey Shallit\\
School of Computer Science\\
University of Waterloo \\
Waterloo, ON  N2L 3G1 \\
Canada\\
\href{mailto:shallit@uwaterloo.ca}{\tt shallit@uwaterloo.ca}
}

\maketitle

\begin{abstract}
In 2013, Fici and Zamboni proved a number of theorems about finite
and infinite words having only a small number of factors
that are palindromes.
In this paper we rederive some of their results, and obtain some new ones,
by a different method based on finite automata.   
\end{abstract}

\section{Introduction}
In this paper we are concerned with certain avoidance properties
of finite and infinite words.

Recall that a word $x$ is said to be a {\it factor} of a word 
$w$ if there exist words $y,z$ such that $w = yxz$.  For example,
the word {\tt act} is a factor of the English word
{\tt factor}.  We sometimes say $w$ {\it contains\/} $x$.
Another term for {\it factor}
is {\it subword}, although this latter term sometimes refers to
a different concept entirely.  We say a (finite or infinite) 
word $x$ {\it avoids} a
set $S$ if no element of $S$ is a factor of $x$.

The reverse of a word $x$ is written $x^R$.   Thus, for example,
${\tt (drawer)}^R = {\tt reward}$.  A word is a {\it palindrome\/}
if $x = x^R$, such as the English word {\tt radar}.  A 
palindrome is called {\it even\/} if its length is even, and
{\it odd\/} if its length is odd.  For example, the English
word {\tt noon} is even, while {\tt madam} is odd.

Fici and Zamboni \cite{Fici&Zamboni:2013} studied avoidance of 
palindromes.   In particular, they were interested in constructing
infinite words with the minimum possible number of distinct
palindromic factors, and
infinite words that minimize the length of the largest palindromic
factor.  In both cases these minima crucially depend on the size of the
underlying alphabet.  

In this paper we revisit their results using a different approach.
The crucial observation is in Section~\ref{two}:  the set of
finite words over a finite alphabet containing at most
$n$ distinct palindromic factors (resp., whose largest palindromic
factor is of length at most $n$) is regular.

The companion paper to this one is \cite{Fleischer&Shallit:2019},
where some of the same ideas are used.

\section{Palindromes and regularity}
\label{two}

Let $x$ be a finite or infinite word.  The set of all of its
factors is written $\Fac(x)$, and the set of its
factors that are palindromes is written $\PalFac(x)$.  
Let $P_\ell (\Sigma)$ (resp., $P_{\leq \ell} (\Sigma)$) 
be the set of all palindromes of length
$\ell$ (resp., length $\leq \ell$) over $\Sigma$.  Of course, since 
both of these sets are finite, they are regular.

\begin{theorem}
Let $S$ be a finite set of palindromes over an alphabet
$\Sigma$.  Then the language
$$ C_\Sigma(S) :=  \{ x \in \Sigma^* \ : \ \PalFac(x) \subseteq S \} $$
is regular.
\label{one}
\end{theorem}

\begin{proof}
Let $\ell$ be the length of the longest palindrome in $S$.  
We claim that $ \overline{C_\Sigma (S)} = L$, where 
$$L = \bigcup_{t \in P_{\leq \ell+2} - S} \Sigma^* \, t \ \Sigma^* . $$

\noindent $\overline{C_\Sigma(S)} \subseteq L$:   If $x \in
\overline{C_\Sigma(S)}$, then $x$ must have some palindromic factor $y$
such that $y \not\in S$.   If $|y| \leq \ell+2$, then
$y \in P_{\leq \ell+2} - S$.  If $|y| > \ell+2$, we can write 
$y = uvu^R$ for some palindrome $v$ such that $|v| \in \{ \ell+1, \ell+2 \}$.
Hence $x$ has the palindromic factor $v$ and $v \in P_{\leq \ell+2} - S$.
In both cases $x \in L$.

\bigskip

\noindent $L \subseteq \overline{C_\Sigma(S)}$:   Let $x \in L$.
Then $x \in  \Sigma^* \, t \ \Sigma^* $ for some
$t \in P_{\leq \ell+2} - S$.  Hence $x$ has a palindromic factor
outside the set $S$ and so $x \not\in C_\Sigma(S)$.

\bigskip

Thus we have written $\overline{C_\Sigma(S)}$ as the finite union of
regular languages, and so $C_\Sigma(S)$ is also regular.
\end{proof}

\begin{remark}
The set $P_{\leq \ell+2} (\Sigma) - S$ can be fairly large.  However, because
$$ \Sigma^* \, x \, \Sigma^* \subseteq \Sigma^* \, y \, \Sigma^*$$
if $y$ is a factor of $x$,
we can replace $P_{\leq \ell+2} (\Sigma) - S$ in
Theorem~\ref{one} with the subset of its minimal elements
under the factor ordering.
(An element $x \in T$ is {\it minimal\/} if $x, y \in T$ with
$y$ a factor of $x$ implies that $x = y$.)  This typically will
have many fewer elements.
\end{remark}

\begin{corollary}
\leavevmode
\begin{itemize}
\item[(a)]
Let $D_{\ell} (\Sigma)$ be the set of finite words over $\Sigma$ containing
at most $\ell$ distinct palindromes as factors
(including the empty word).  Then $D_{\ell} (\Sigma)$ is regular.

\item[(b)]
Let $E_{\ell} (\Sigma)$ be the set of finite words over
$\Sigma$ containing no palindrome of length $> \ell$ as a factor.
Then $E_{\ell} (\Sigma)$ is regular.

\item[(c)]
Let $R_{\ell,m} (\Sigma)$ be the set of finite words over
$\Sigma$ containing no even palindrome of length $>\ell$ nor
any odd palindrome of length $>m$ as factors.  Then
$R_{\ell,m}(\Sigma)$ is regular.

\item[(d)]
Let $T_{\ell,m} (\Sigma)$ be the set of finite words over
$\Sigma$ containing at most $\ell$ even palindromes and
$m$ odd palindromes.  Then $T_{\ell,m}(\Sigma)$ is regular.
\end{itemize}
\label{palreg}
\end{corollary}

\begin{proof}
\leavevmode
\begin{itemize}
\item[(a)] Note that no word in $D_{\ell} (\Sigma)$ cannot contain a palindrome of
length $r \geq 2\ell$ as a factor, because then it would also contain
palindromes of length $0, 2,\ldots r-2$ as factors ($r$ even)
or length $1,3, \ldots r-2$ as factors ($r$ odd).  In both
cases this gives at least $\ell+1$ distinct palindromes.

Hence 
$$ D_\ell (\Sigma) = 
\bigcup_{
{|S| \leq \ell}
\atop
{S \subseteq P_{\leq 2\ell-1} (\Sigma)} 
}
C_{\Sigma} (S) , $$
the union of a finite number of regular languages.

\item[(b)]  We have $E_\ell (\Sigma) = C_{\Sigma} (P_{\leq \ell} (\Sigma)) $.

\item[(c)]   We have $R_{\ell,m} (\Sigma) =
C_{\Sigma} \biggl( \bigl(P_{\leq \ell} (\Sigma) 
\, \cap \, (\Sigma^2)^* \bigr) \ \cup \ \bigl(P_{\leq m} (\Sigma) \, \cap \, \Sigma(\Sigma^2)^* \bigr) \biggr) $.

\item[(d)]  We have $$T_{\ell, m} (\Sigma) =
\bigcup_{{|S_1| \leq \ell} \atop {S_1 \subseteq \bigcup_{0 \leq i < \ell}
	P_{2i} (\Sigma)}} C_\Sigma (S_1) \ \cup \ 
\bigcup_{{|S_2| \leq m} \atop {S_1 \subseteq \bigcup_{0 \leq i < m}
	P_{2i+1} (\Sigma)}} C_\Sigma (S_2).$$

\end{itemize}
\end{proof}

Theorem~\ref{one} and Corollary~\ref{palreg} implicitly provide
an algorithm for actually finding the DFA's accepting the languages
$D_\ell (\Sigma)$, $E_\ell (\Sigma)$, $R_{\ell,m} (\Sigma)$, 
and $T_{\ell,m} (\Sigma)$:  namely,
construct automata for each term of the unions and intersections,
and combine them
using standard techniques (e.g., \cite[Sect.~3.2]{Hopcroft&Ullman:1979}),
possibly using minimization at each step.  This can be
carried out, for example, using a software package such as
{\tt Grail} \cite{Raymond&Wood:1994,Campeanu:2019}.

However, our experience shows that the intermediate automata so generated
can be quite large.  Instead, we use a different approach to
construct the automata directly, which
we now illustrate for the case of $D_\ell(\Sigma)$, as follows.

The states are of the form $\Sigma^{\leq 2\ell-1} \times 2^U$, where
$U$ is the set of the nonempty palindromes of length at most $2\ell-1$.
Given a state of the form $(x, S)$, upon reading the letter $a$,
we go to the new state $(y, T)$, where $y = xa$ (if $|xa| \leq 2\ell-1$)
or the suffix of length $2\ell-1$ of $xa$ (if $|xa| = 2\ell$), and
$T = S \ \cup \ \PalFac(xa)$.  If $|T| > \ell$, it is labeled as
a rejecting state.

The resulting automaton, as described, still can be rather large.  However,
many states will not be reachable from the start state.  Instead,
we construct all reachable states
using a queue, in a breadth-first manner starting from the
initial state $(\varepsilon, \emptyset)$.  As soon as we reach a state
$(x,S)$ with $|S| > \ell$, the state is labeled as a dead state and
we do not append it to the queue.

We implemented this idea in Dyalog APL.  Our program creates an
automaton in {\tt Grail} format which can then be minimized using
{\tt Grail}.

Our approach allows us to recover many of the results
of Fici and Zamboni, and even more.   For example, the DFA's
we compute give us
a complete description of {\it all\/} words, both finite
and infinite, containing at most $\ell$ distinct palindromic factors.
It provides an easy and efficient way to determine whether
or not there exist infinite words containing a given avoidance property,
and if so, whether some of these words are aperiodic.
As corollaries, we can computably determine a linear recurrence
giving the number $a(n)$ of such words of length $n$,
and the asymptotic growth rate of the sequence $(a(n))_{n \geq 0}$.  

Finally, our approach replaces a long case-based argument
that can be difficult to follow, and is prone to error,
with a machine computation that can be verified mechanically.

\section{Linear recurrences and automata}
\label{three}

We summarize some well-known techniques for enumerating the
number of length-$n$ words accepted by deterministic finite
automata that we use in this paper.  For more details,
see, for example, \cite[Sect.~3.8]{Shallit:2009} and
\cite{Everest&vdp&Shparlinski&Ward:2003}.

We introduce some notation and terminology:  if
$q(X) = q_t X^t + q_{t-1} X^{t-1} + \cdots + q_1X + q_0$ 
is a polynomial and ${\bf a} = (a(n))_{n \geq 0}$ is a sequence,
then $q \circ {\bf a}$ denotes the sequence 
$(q_t a(t+i) + q_{t-1} a(t+i-1) + \cdots
+ q_1 a(i+1) + q_0 a(i))_{i \geq 0}$ obtained by taking the
dot product of the coefficients of $q$ with sliding ``windows''
of the sequence $\bf a$.
If $q \circ {\bf a}$ is the sequence $(0,0,0,\ldots)$, we 
call $q$ an {\it annihilator\/} of $(a(n))_{n \geq 0}$.
It is now easy to verify that if $q, r$ are polynomials, then
$(qr) \circ {\bf a} = q \circ (r \circ {\bf a})$.
We also define $\Lead(q) = q_t$ to be the leading coefficient of $q$.

Suppose $Q = \{ q_0, q_1, \ldots, q_{r-1} \}$ and
$A = (Q, \Sigma, \delta, q_0, F)$ is an $r$-state DFA.
From this we can compute an $n \times n$ matrix
$M$ such that $M[i,j] = \{ a \in \Sigma \ : \ \delta(q_i, a) = q_j \}$.
Let $v = [1 \ 0 \ 0 \ \cdots 0]$ be the row vector
with a $1$ in the first position and $0$'s elsewhere, and
let $w$ be the column vector with $1$'s in positions corresponding
to the final states $F$ and $0$'s corresponding to $Q-F$.
Then $a(n)$, the number of length-$n$ words accepted by $A$, is
$v M^n w$.

We can find a linear recurrence for the sequence $(a(n))_{n \geq 0}$ as follows:
first, we compute the minimal polynomial $p(X) =
X^t + p_{t-1} X^{t-1} + \cdots + p_1X + p_0$ of $M$ using standard
techniques.   Then $p(M) = 0$, so
$M^t + p_{t-1} M^{t-1} + \cdots + p_1M + p_0I = 0$.  By multiplying
by $M^i$, we get
$M^{t+i} + p_{t-1} M^{t+i-1} + \cdots + p_1M^{i+1} + p_0 M^i = 0$.
By premultiplication by $v$ and postmultiplication by $w$, we get
$v M^{t+i} w + p_{t-1} v M^{t+i-1} w + \cdots
+ p_1 v M^{i+1} w + p_0 vM^i w = 0$.
Hence $a(t+i) + p_{t-1} a(t+i-1) + \cdots + p_1 a(i+1) + p_0 a(i) = 0$,
and hence $(a(n))_{n \geq 0}$ satisfies a linear
recurrence with constant coefficients given by the $p_i$.
Using our terminology, the polynomial $p$ annihilates $(a(n))_{n \geq 0}$.

However, $p$ may not be the lowest-degree
annihilator of $(a(n))_{n \geq 0}$.  A
lower degree annihilator will necessarily be a divisor of the polynomial
$p$.  The lowest degree annihilator can be determined using an algorithm
based on the following theorem, which seems to be new.

\begin{theorem}
Suppose the polynomial $p(X)$, with leading coefficient nonzero,
annihilates the sequence $(a(n))_{n \geq 0}$ and suppose $q(x) \, | \, p(x)$.
If the polynomial ${p \over q}$ also
annihilates the sequence $(a(n))_{n \geq 0}$ for the
first $\deg q$ consecutive windows
of $(a(n))_{n \geq 0}$, 
then it annihilates all of $(a(n))_{n \geq 0}$.
\end{theorem}

\begin{proof}
Suppose $p \over q$ annihilates ${\bf a} = (a(n))_{n \geq 0}$
for the first $s := \deg q$ consecutive windows of $\bf a$, but
not all of $\bf a$.

Write ${p \over q} = d_t X^t + \cdots + d_1 X + d_0$.  Define
$(f(n))_{n \geq 0} = {p \over q} \circ (a(n))_{n \geq 0}$.
Thus $ f(n) = \sum_{0 \leq i \leq t} d_i a(n+i)$.  Then by hypothesis
we have $f(n) = 0$ for $n = 0, 1, \ldots , s-1$.  Now $p$
annihilates $(a(n))_{n \geq 0}$, so $q$ annihilates $(f(n))_{n \geq 0}$.
Let $r$ be the least index such that
$f(r) \not= 0$.  So $(f(0), f(1), \ldots, f(r)) = (0,0, \ldots ,0, e)$
for some $e \not= 0$.  But $q$ annihilates $(f(n))_{n \geq 0}$,
so if $r \geq s$ then
$q \circ ( \overbrace{0,0, \ldots , 0}^{s-1}, e) = 0$.  But
$q \circ ( \overbrace{0,0,\ldots ,0}^{s-1} ,e) = e \Lead(q) \not= 0$,
a contradiction. 
\end{proof}

This gives us the following algorithm for finding the lowest-degree
annihilator of a recurrence.

\bigskip\hrule
\begin{tabbing} 
{\sc Algorithm LDA}$(p, {\bf a})$ \\
\ \\
Write $p := q_1 q_2 \cdots q_m$, the product of (not necessarily
distinct) irreducible factors.  \\
For \= $i := 1$ to $m$ do \\
\> 	$r := p/q_i$ \\
\>	If $r$ annihilates the first $\deg r$ windows of $\bf a$, set
$p := p/q_i$. \\
return($p$);
\end{tabbing}
\smallskip\hrule
\bigskip

Terms of the form $X^n$ in an annihilator can be removed if one assumes
that the recurrence begins at $a(n)$ instead of $a(0)$.  For this reason,
in this paper, we do not report such terms in our annihilators.

In our computations, we used {\tt Maple} to compute minimal polynomials
(via the {\tt LinearAlgebra} package) and factor them.

\section{Automata and infinite words}

We recall some material from the companion paper \cite{Fleischer&Shallit:2019}.

The DFA's generated in this paper are for regular
languages $L$ that are defined by avoidance of a
finite set $S$ of finite words.  Such languages are
called {\it factorial}; that is,
every factor of a word of $L$ is also a word of $L$.

The minimal DFA $M = (Q, 
\Sigma, \delta, q_0, F)$ for a factorial language $L \not= \Sigma^*$
has exactly one nonaccepting state,
which is the dead state.  (A state is {\it dead\/} if it is nonaccepting
and transitions to itself on all letters of the alphabet $\Sigma$.)
{\it In this paper, we do not display
this dead state in our figures, nor count it in our
discussion of the cardinality of a DFA's states.}

The (one-sided) infinite words with the given avoidance property are
then given by the infinite paths through $M$, starting at the start
state $q_0$.  

A state $q$ is called {\it recurrent} if there is a nonempty word
$w$ such that $\delta(q,w) = q$.  A state $q$ is called
{\it birecurrent} if there are two noncommuting words $x_0, x_1$ such
that $\delta(q,x_0) = \delta(q,x_1) = q$.

As shown in \cite{Fleischer&Shallit:2019}, an infinite word having
the desired avoidance property exists iff $M$ has a recurrent state,
and aperiodic infinite words exist iff $M$ has a birecurrent state.
In this latter case, there are actually uncountably many such words.
As shown in \cite{Fleischer&Shallit:2019}, these correspond to
the image, under the morphism $h:  0 \rightarrow x_0$ and $1 \rightarrow x_1$
of an aperiodic binary word.

Furthermore, we can find infinite words avoiding $S$ that are
(a) uniformly recurrent and aperiodic (b) linearly recurrent and
aperiodic and (c) $k$-automatic for any $k \geq 2$ and uniformly
recurrent and aperiodic
and (d) the fixed point of a primitive uniform morphism, which is
uniformly recurrent.

To see this, note that the image under a nonerasing morphism of a
uniformly recurrent infinite word is uniformly recurrent.  So it
suffices to apply $h$ to any uniformly recurrent binary word, such
as the Thue-Morse word $\bf t$ \cite{Allouche&Shallit:1999}.

Similarly, the image under a nonerasing morphism of a linearly
recurrent infinite word is linearly recurrent.

To see that we can find a $k$-automatic word with the desired 
properties, note that we can start with any $k$-automatic word
that is uniformly recurrent and aperiodic (for example,
the fixed point of $0 \rightarrow 0^{k-1} 1$ and
$ 1 \rightarrow 1 0^{k-1}$) and apply the morphism $h$
to it.

Finally, assume that $x_0$ and $x_1$ are chosen such that
for some $a\in \{ 0, 1\}$ we have $x_a$ starts with $a$.
Let $b = \{ 0, 1 \} - \{a\}$.  Write $g(a) = x_a x_b$ and
$g(b) = x_b x_a$.
Then $g^\omega(a)$, the infinite fixed point of $g$ starting with
$a$, is uniformly recurrent.

In what follows, we use the alphabet $\Sigma_k = \{ 0, 1, \ldots, k-1 \}$.
A-numbers in the paper refer to sequences from the {\it On-Line Encyclopedia of
Integer Sequences} \cite{Sloane:2019}.

\section{Minimizing the number of palindromes}

We define $d_{k,\ell} (n)$ to be the number of length-$n$ words
in $D_\ell(\Sigma_k)$.

\subsection{Alphabet size 2}

\begin{theorem}{(Fici-Zamboni)}
There are infinite binary words containing at most 9 palindromes.
All are periodic, and of the form $x^\omega$ for
$x$ a conjugate of either $001011$ or $001101$.
There are no infinite binary words containing at most 8 palindromes.
\label{thm2-9}
\end{theorem}

\begin{proof}
We construct the DFA for $D_9 (\Sigma_2)$
as in Section~\ref{two}.  It has 611 states before
minimization and 98 after minimization, and we omit it
here.  No state is birecurrent,
but there are 12 recurrent states.   Examining the associated
paths easily gives the result.

To see the result for 8 palindromes, we can construct the 
DFA for $D_8(\Sigma_2)$.
It has 259 states before minimization and 23 after minimization.
No state is recurrent.  The longest word accepted is of length $8$.
Alternatively, one can prove this result using a simple breadth-first
search of the space of words.
\end{proof}

\begin{theorem}{(Restatement of Fici-Zamboni)}
There are exactly 40 infinite binary words containing exactly 10 palindromes.
All are ultimately periodic, and are of the following forms:
\begin{itemize}
\item $x^\omega$ for $x$ a conjugate of $0001011$, $0001101$, $0010111$, 
or $0011101$;
\item $y (001011)^\omega$ for $y \in \{ 0, 01, 111, 0011, 11011, 101011 \}$;
\item $y (001101)^\omega$ for $y \in \{ 0, 11, 001, 0101, 11101, 101101 \}$.
\end{itemize}
\end{theorem}

\begin{proof}
We create the automaton for $D_{10} (\Sigma_2)$ as in Section~\ref{two}.
It has 1655 states before minimization and 280 after.  None of these
states are birecurrent.  By examining the possible infinite paths, we see
these include those of Theorem~\ref{thm2-9} and the ones listed above.
\end{proof}

\begin{theorem}
There are uncountably many aperiodic, uniformly
recurrent infinite binary words
containing exactly 11 palindromes.  
\end{theorem}

\begin{proof}
Using the method in Section~\ref{two}, we can construct the DFA
for $D_{11} (\Sigma_2)$.  It has 5253 states before
minimization, and 810 states afterwards, for
$D_{11} (\Sigma_2)$.
We do not give the latter automaton here, as it is
too large to display in a reasonable way, but it can be downloaded
from the second author's website at

State 738 is birecurrent, with
two paths labeled $x_0 = 0001011001011$ and
$x_1 = 001011001011$.  
\end{proof}

\begin{corollary}
The number of binary words containing at most 11 distinct palindromic
factors (including the empty word) is $(d_{2,11} (n))_{n \geq 0}$, where
\begin{displaymath}
\begin{split}
&(d_{2,11}(0), \ldots, d_{2,11}(41)) = (1,2,4,8,16,32,64,128,256,512,1024,292,270,268,
276,276,288, \\
& 320,340, 364, 388,404,428,476,512,560,610,644,692,768,840,924, 
1020,1100,1190,1316, \\
& 1452,1612, 1786,1952,2134,2348)
\end{split}
\end{displaymath}
and 
\begin{align*}
d_{2,11}(n) &= -d_{2,11}({n-1}) - d_{2,11}({n-2}) - d_{2,11}({n-3})
-d_{2,11}({n-4}) - d_{2,11}({n-5}) + 2d_{2,11}({n-6})  \\
&\quad + 4d_{2,11}({n-7}) + 5d_{2,11}({n-8}) + 5 d_{2,11}({n-9}) 
+ 5 d_{2,11}({n-10}) + 5 d_{2,11}({n-11})   \\
& \quad + 2 d_{2,11}({n-12}) -3 d_{2,11}({n-13}) + -6d_{2,11}({n-14})
-8 d_{2,11}({n-15}) -8 d_{2,11}({n-16}) \\
& \quad -8 d_{2,11}({n-17}) -7 d_{2,11}({n-18}) - 3 d_{2,11}({n-19}) +
3 d_{2,11}({n-21}) + 4d_{2,11}({n-22}) \\
& \quad + 4 d_{2,11}({n-23}) + 4 d_{2,11}({n-24})
+ 3 d_{2,11}({n-25}) +2 d_{2,11}({n-26}) + d_{2,11}({n-27})
\end{align*}
for $n \geq 42$.

Asymptotically, $d_{2,11}(n) \sim c \cdot \alpha^n$, where
$\alpha \doteq 1.1127756842787054706297$ is the largest positive
real zero of $X^7 - X - 1$ and
$c \doteq 20.665$.
\end{corollary}

\begin{proof}
Using {\tt Maple}, we computed the minimal polynomial for the
matrix of the 811-state DFA described above.  It is
\begin{multline*}
X^{15} (X-1)(X-2)(X+1)(X^2 + 1)(X^2 + X + 1)(X^2 - X + 1)(X^7 - X - 1) \\
(X^4 + 1)(X^6 + X^5 + X^4 + X^3 + X^2 + X + 1)(X^8 - X^2 - 1) .
\end{multline*}
Next, using the procedure described in Section~\ref{three},
we can find the minimal annihilator of the recurrence.
It is
\begin{equation}
(X-1)(X+1)(X^2+X+1)(X^2-X+1)(X^7-X-1)(X^6 + X^5 + X^4 + X^3 + X^2 + X + 1)(X^8 - X^2 - 1).
\label{polys}
\end{equation}
When expanded, this gives the coefficients of the annhiliator of
the sequence $(d_{2,11}(n))_{n \geq 0}$, which are given above.

To get the asymptotic behavior of the recurrence, we must find the largest
real zero of the polynomials given in \eqref{polys}.  It is the largest
real zero of $X^7-X-1$, which is approximately $1.1127756842787054706297$.
\end{proof}

\begin{remark}
This is sequence \seqnum{A330127} in the OEIS.
\end{remark}

In their paper,
Fici and Zamboni constructed a uniformly recurrent aperiodic binary word
containing 13 palindromic factors, and
whose set of factors is closed under reversal.  We achieve the same
result using a different construction and a different proof.

\begin{theorem}
Define $G_0 = 001101000110$ and $G_{n+1} = G_n 01 G_n^R$ for $n \geq 0$.
Then $G_{\infty} = \lim_{n \rightarrow \infty} G_n$
is uniformly recurrent, aperiodic, and has 13 palindromic factors.
\end{theorem}

\begin{proof}
We start by constructing the DFA for the language $D_{13}(\Sigma_2)$
using the method described in Section~\ref{two}.
This DFA $M$ has 93125 states before
minimization and 6522 states after minimization.  The unique
dead state is numbered 3012.

Next, we look at the transformations $\tau_n$ of states induced by the
words $G_n$.   We claim that
\begin{itemize}
\item $\tau_{G_n} = \tau_{G_{n+1}}$ for $n \geq 2$;
\item $\tau_{G_n} = \tau_{G_n^R}$ for $n \geq 1$.
\end{itemize}
which can be easily verified by induction using the transition
function for $M$.  

The resulting transformations of states for $n \geq 2$ are as follows:
\begin{align*}
0 &  \xrightarrow{\ G_n\ } 4882 \xrightarrow{\ 01 \ } 5058
 \xrightarrow{\ G_n^R\ } 4882  \\
0 & \xrightarrow{\ G_n^R\ } 4882 \xrightarrow{\ 10 \ } 5059 
 \xrightarrow{\ G_n\ } 4882
\end{align*}
Since these paths do not end in the unique nonaccepting state, the 
corresponding words contain at most $13$ palindromes.

It is easy to see that the word $G_\infty$ is uniformly recurrent and
closed under reversal.  This is left to the reader.

The fact that $G_\infty$ is not ultimately periodic follows
from \cite[Thm.~4]{Shallit:1982b}.
\end{proof}

\subsection{Alphabet size 3}

\begin{theorem}{(Fici-Zamboni)}
If a ternary infinite word contains $4$ palindromes (including
the empty word), it is necessarily of the form $(abc)^\omega$
for distinct letters $a,b, c$.  No ternary infinite word can contain
$3$ or fewer palindromes.
\end{theorem}

\begin{proof}
We construct the DFA for $D_4(\Sigma_3)$
using the algorithm suggested in
Section~\ref{two}.  It has 52 states and,
when minimized, has 18 states.
It is depicted below in Figure~\ref{pal3num4}.  
Only the states numbered $12, 13, 14, 15, 16, 17$ are recurrent, and
none of them are birecurrent.  The desired result now easily follows from 
examining the possible paths through these states.

\begin{figure}[H]
\begin{center}
\includegraphics[width=5.5in]{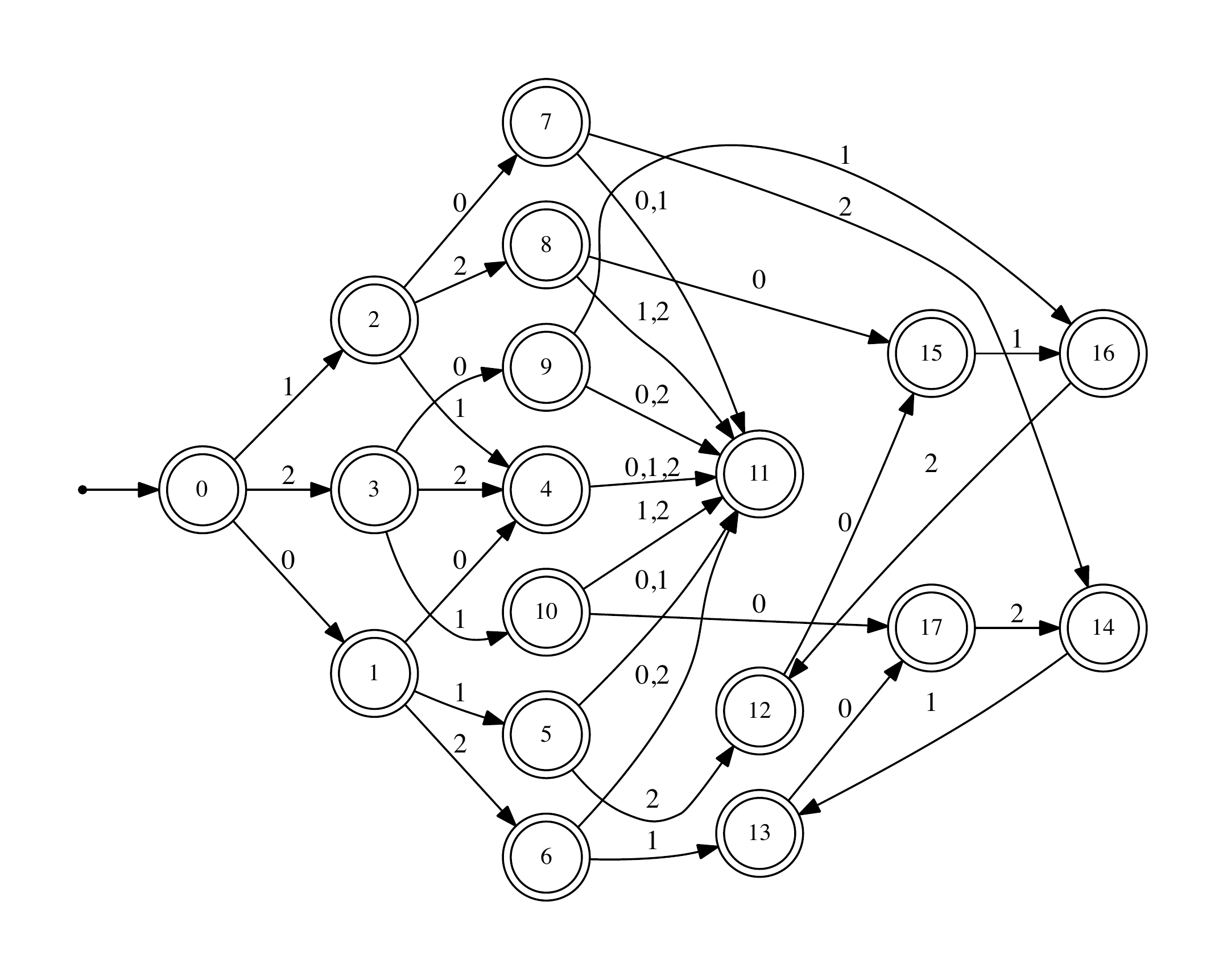}
\end{center}
\caption{Automaton for ternary words containing at most $4$ palindromes}
\label{pal3num4}
\end{figure}

To see that no ternary infinite word can contain $3$ or fewer palindromes,
we can perform the same construction as above, but for $3$ palindromes.
The resulting automaton has 13
states (3 when minimized) and no recurrent states.  We omit it here.
Alternatively, one can prove this result with a simple breadth-first
search of the space of words.
\end{proof}

\begin{theorem}
There are uncountably many aperiodic ternary words containing at most
$5$ palindromic factors.
\end{theorem}

\begin{proof}
We can construct the automaton for $D_5(\Sigma_3)$
as described in Section~\ref{two}.
It has 319 states before minimization and 69 states after.
We do not depict it here, as it is too large to visualize
clearly.   The state 39 is birecurrent, with paths labeled
$x_0 = 0012$ and $x_1 = 012$.
\end{proof}

\begin{corollary}
The number of ternary words containing at most $5$ palindromic
factors is $d_{3,5}(n) $, where
$(d_{3,5}(0), \ldots, d_{3,5}(8)) = (1,3,9,27,81,42,54,66,78)$ and
$d_{3,5}(n) = d_{3,5}(n-3) + d_{3,5}(n-4)$ for $n \geq 9$.
Asymptotically we have $d_{3,5}(n) \sim c \alpha^n$ where
$\alpha \doteq 1.2207440846$ and
$c \doteq 16.07007$.
\end{corollary}

\begin{proof}
The minimal polynomial of the corresponding matrix
is
$$ X^5 (X-1) (X-3) (X^2 + X + 1) (X^4 - X - 1) .$$
Using the method in Section~\ref{three}, we can find the minimal 
annihilator of the sequence, which is $X^4 - X - 1$.
The result now follows.
\end{proof}

\begin{remark}
This is sequence \seqnum{A329023} in the OEIS.
We have $d_{3,5}(n) = 6 \cdot$\seqnum{A164317}$(n)$ for $n \geq 5$.
\end{remark}

\section{Lengths of palindromes}

Instead of minimizing the total number of palindromes, 
Fici and Zamboni also considered minimizing the length of the
longest palindrome.  We can also do that with our method.

We define $e_{k,\ell} (n)$ to be the number of length-$n$ words
in $E_\ell (\Sigma_k)$.

\subsection{Alphabet size 2}

\begin{theorem}{(Restatement of Fici-Zamboni)}
There are exactly 20 infinite binary words having no palindromes
of length $>4$, and all are ultimately periodic.  They are 
as follows:
\begin{itemize}
\item $x^\omega$ for $x$ a conjugate of $001011$;
\item $x^\omega$ for $x$ a conjugate of $001101$;
\item $(0+00+111+1111)(001011)^\omega$;
\item $(0+00+11101+111101)(001101)^\omega$.
\end{itemize}
\label{pal2len4-thm}
\end{theorem}

\begin{proof}
The automaton for $E_4(\Sigma_2)$ is depicted in Figure~\ref{pal2len4},
and the only infinite paths are those given.  (There are no
birecurrent states.)
\end{proof}

\begin{figure}[H]
\begin{center}
\includegraphics[width=6.5in]{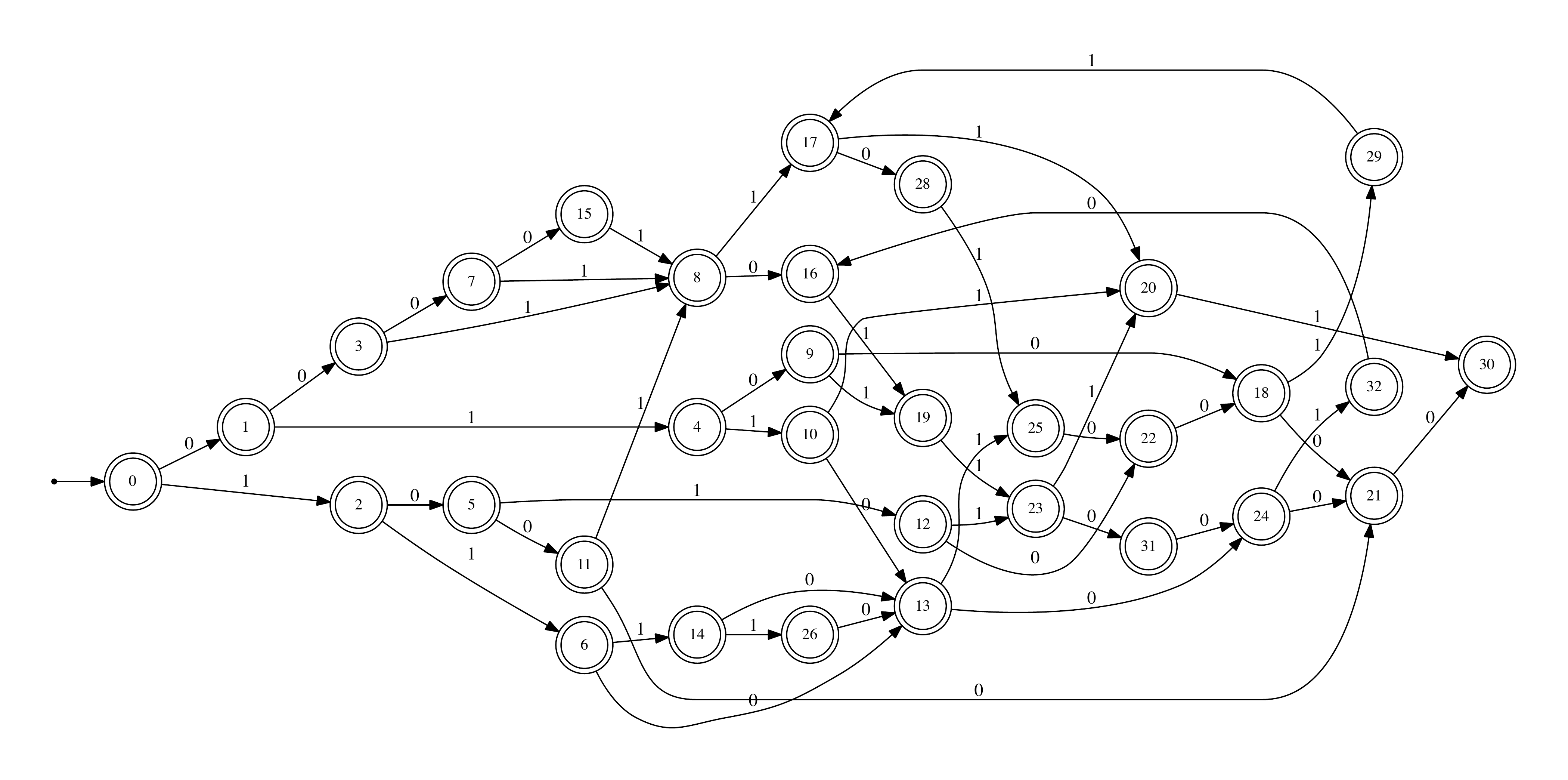}
\end{center}
\caption{Automaton for binary words containing no palindromes of
length $>4$}
\label{pal2len4}
\end{figure}

\begin{theorem}
There are uncountably many uniformly recurrent
binary words containing no palindromes of length $>5$.
They are the labels of the paths through the automaton
in Figure~\ref{pal2len5}.
\end{theorem}

\begin{figure}[H]
\begin{center}
\includegraphics[width=6.5in]{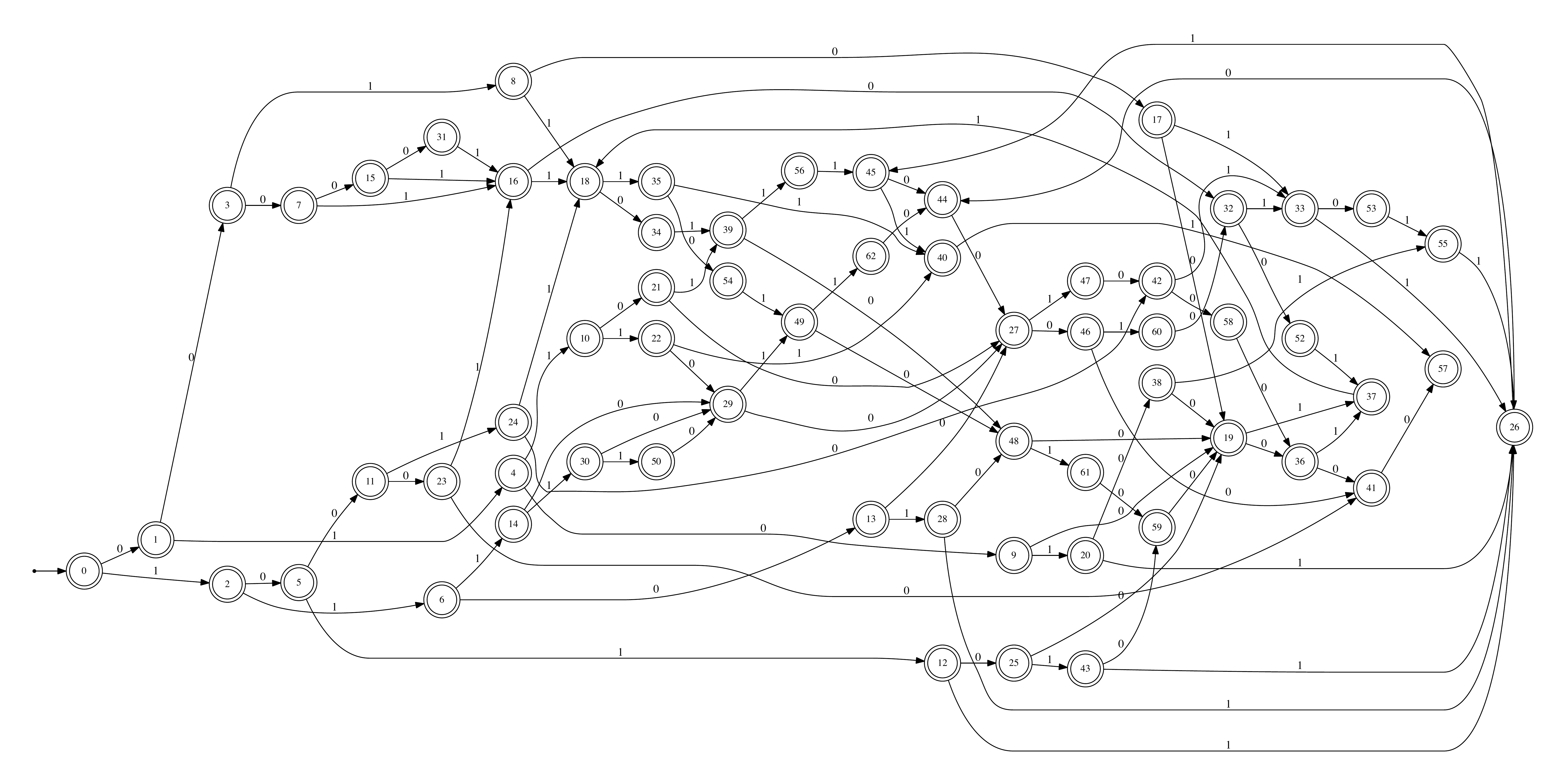}
\end{center}
\caption{Automaton for binary words containing no palindromes of
length $>5$}
\label{pal2len5}
\end{figure}

\begin{proof}
As before.  There are 719 states in the unminimized automaton 
for $E_5 (\Sigma_2)$
and 62 states in the minimized one. 
State 44 is birecurrent, with paths
$x_0 = 01010110$ and $x_1 = 0010101110$.
\end{proof}

\begin{theorem}
The sequence $(e_{2,5} (n))_{n \geq 0}$ counting the number of binary words of length $n$
containing no palindromes of length $>5$ satisfies the recurrence
$$e_{2,5}(n) = 3 e_{2,5}(n-6) + 2 e_{2,5}(n-7) + 2 e_{2,5}(n-8) + 2 e_{2,5}(n-9) + e_{2,5}(n-10)$$
for $n \geq 20$.  Asymptotically
$e_{2,5}(n) \sim c \alpha^n$ where $\alpha \doteq 
1.36927381628918060784\cdots$ is the positive real zero of the equation 
$X^{10}-3X^4-2X^3-2X^2-2X-1$, and $c = 9.8315779\cdots$.
\end{theorem}

\begin{proof}
The minimal polynomial of the corresponding matrix is
\begin{equation}
 X^{10} (X-2) (X^{10} + X^4 - 2X^3 -2X^2 -2X - 1)(X^{10} -3X^4 -2X^3 -2X^2 -2X - 1).
 \label{pal35}
\end{equation}
The technique described in Theorem~\ref{three} can be used to
find the minimal annihilator for the recurrence.  It is the last term
in the factorization \eqref{pal35}.
\end{proof}

\begin{remark}
The sequence $e_{2,5} (n)$ is sequence \seqnum{A329824} in the OEIS.
\end{remark}

\subsection{Alphabet size 3}

\begin{theorem}{(Fici-Zamboni)}
The only infinite ternary words having no palindromes of length
$>1$ are those of the form $(abc)^\omega$ for distinct letters
$a,b,c$.    
\end{theorem}

\begin{proof}
The automaton for $E_1(\Sigma_3)$
has 16 states before minimization and 10 states after.
We omit it here.  There are no birecurrent states, and
the only infinite paths are those given.
\end{proof}

\begin{theorem}
There are uncountably many ternary words containing no palindromes
of length $>2$.
\label{pal3len2-thm}
\end{theorem}

\begin{proof}
We can construct the automaton for
$E_2(\Sigma_3)$ as in Section~\ref{two}.  It has 67 states
unminimized and 19 states when minimized.  
State $6$ is birecurrent, with paths labeled
$x_0 = 211002$ and $x_1 = 11002$.

\begin{figure}[H]
\begin{center}
\includegraphics[width=6in]{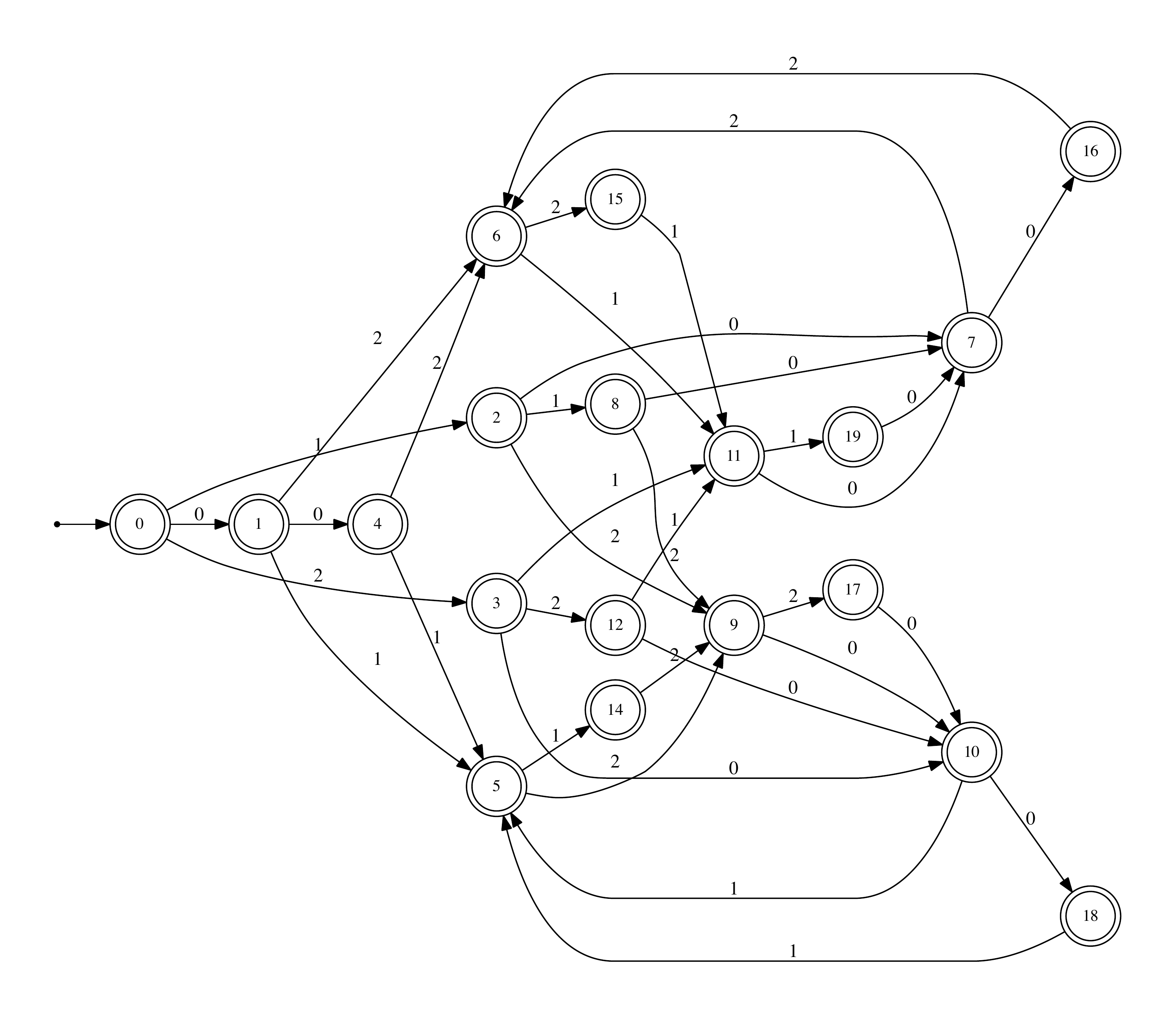}
\end{center}
\caption{Automaton for ternary words containing no palindromes of
length $>2$}
\label{pal3len2}
\end{figure}
\end{proof}

The Fibonacci numbers $F_n$ are defined by $F_0 = 0$ and $F_1 = 1$ and
$F_n = F_{n-1} + F_{n-2}$.

\begin{corollary}
The number $e_{3,2}(n)$ of length-$n$ ternary words containing no
palindromes of length $>2$ is $6F_{n+1}$ for $n \geq 3$.  
\end{corollary}

\begin{proof}
The minimal polynomial of the matrix is $X^3 (X-3)(X^2 - X - 1)(X^4 +
X^3 + 2X^2 + 2X + 1)$.  The minimal annihilator is $X^2 - X - 1$.
The result now follows easily.
\end{proof}

\begin{remark}
The sequence $e_{3,2}(n)$
is sequence \seqnum{A330010} in the OEIS.
\end{remark}

\subsection{Alphabet size 4}

Fici and Zamboni proved that, over the alphabet
$\Sigma_4$, there is an infinite aperiodic 
uniformly recurrent word whose only palindromes
are $\varepsilon, 0, 1, 2, 3$.
We show how to handle this
using our method.

\begin{theorem}
There is an infinite aperiodic uniformly recurrent word over
$\Sigma_4$ whose only palindromes are
$\varepsilon, 0, 1, 2, 3$.
\end{theorem}

\begin{proof}
To find the words avoiding all palindromes as factors except
these $5$, we can use Theorem~\ref{one}.  After computing
the minimal elements, it suffices to avoid the factors
$$ \{00,11,22,33,010,020,030,101,121,131,202,212,232,303,313,323\} .$$
The minimal DFA is depicted in Figure~\ref{fig1}.

\begin{figure}[H]
\begin{center}
\includegraphics[width=6.5in]{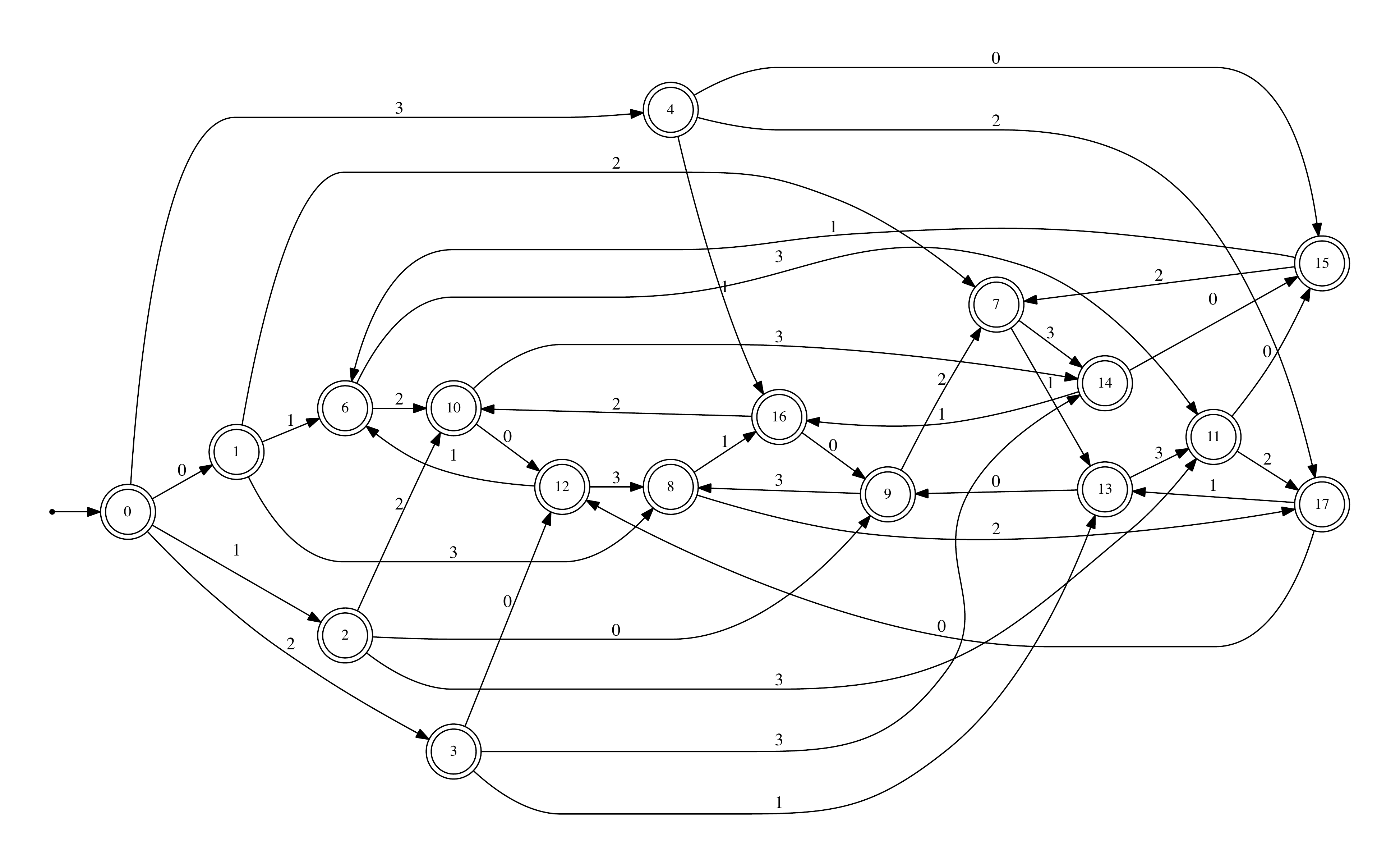}
\end{center}
\caption{Automaton for $4$-letter alphabet.  The dead state, numbered
5, is omitted.}
\label{fig1}
\end{figure}

The state numbered $6$ is birecurrent, with two paths labeled
$2301$ and $301$.  Let ${\bf x}$ be an aperiodic uniformly recurrent
word over $\{ 0, 1\}$ and define the morphism
$h(0) = 2301$ and $h(1) = 301$.  
For example, we can take $\bf x$ to be the Thue-Morse word.
Then $h({\bf x})$ has the desired properties.
\end{proof}

\begin{corollary}
The number $e_{4,1}(n)$ of finite words over $\Sigma_4$ having all their  palindromic
factors contained in $\{ \varepsilon, 0, 1, 2, 3 \}$ is $3 \cdot 2^n$
for $n \geq 2$.
\end{corollary}

\begin{proof}
The minimal polynomial of the matrix corresponding to the
automaton is
$X^2 (X-1) (X-2)(X-4)(X+1)(X^2 + X + 2)$.
Using the procedure in Section~\ref{three} we can determine
the minimal annihilator,which is $X-2$.  
It follows that $e_{4,1} (n) = 3 \cdot 2^n$ for $n \geq 2$.
\end{proof}

Berstel, Boasson, Carton, and Fagnot
\cite{Berstel&Boasson&Carton&Fagnot:2009}
constructed an infinite word over $\Sigma_4$ that is uniformly
recurrent,
has exactly $5$ palindromic factors, and
further is closed under reversal,
as follows:  define $B_0 = 01$ and $B_{n+1} = B_n 23 B_n^R$.  This
is an example of {\it perturbed symmetry}; see
\cite{Dekking&MendesFrance&vanderPoorten:1982} for more details.
We can verify their construction using our method.  Consider the DFA
in Figure~\ref{fig1}; then each word $w$ induces a transformation $\tau_w$ of
the states given by $q \rightarrow \delta(q,w)$.  We claim that
\begin{itemize}
\item[(a)] $\tau_{B_n} = \tau_{B_n^R} = (9, 5, 5, 9, 9, 5, 5, 5, 5, 5, 9, 9, 5, 5, 9, 5, 5, 9)$ for $n \geq 1$;
\item[(b)] $\tau_{23} = (17, 17, 17, 5, 5, 5, 17, 5, 5, 17, 5, 5, 17, 17, 5, 5, 5, 5)$.
\item[(c)] $\tau_{32} = (14, 14, 14, 5, 5, 5, 14, 5, 5, 14, 5, 5, 5, 5, 5, 14, 14, 5)$.
\end{itemize}
The claims about $\tau_{B_1}$, $\tau_{B_1^R}$, $\tau_{23}$, and
$\tau{32}$ are easily verified.  We now prove the claim about $B_n$
by induction.  The reader can now
check that $\tau_{B_{n+1}} = \tau_{B_n 23 B_n^R} = \tau_{B_n}$
and $\tau_{B_{n+1}^R} = \tau_{B_n 32 B_n^R} = \tau_{B_n}$.
Since $0$ is mapped to accepting state $9$ by $B_n$, it follows that
each $B_n$ has the desired properties.

\section{Odd and even palindromes}

In order to illustrate that the technique in this paper has wider
applicability, we now turn to a topic not covered in the paper of 
Fici and Zamboni.
Because an odd palindrome factor of length $\ell$ implies the
existence of odd palindrome factors of all shorter lengths,
and the same for even palindrome factors, it makes sense to
consider minimizing the lengths of odd and even palindrome factors
separately.  This is what we do in this section.

We define $r_{k,\ell, m} (n)$ to be the number of length-$n$ words
in $R_{\ell,m} (\Sigma_k)$.

\subsection{Alphabet size 2}

\begin{theorem}
There are uncountably many uniformly recurrent binary words 
having longest even palindrome factor of length $\leq 2$
and longest odd palindrome of length $\leq 5$.
\end{theorem}

\begin{figure}[H]
\begin{center}
\includegraphics[width=6.5in]{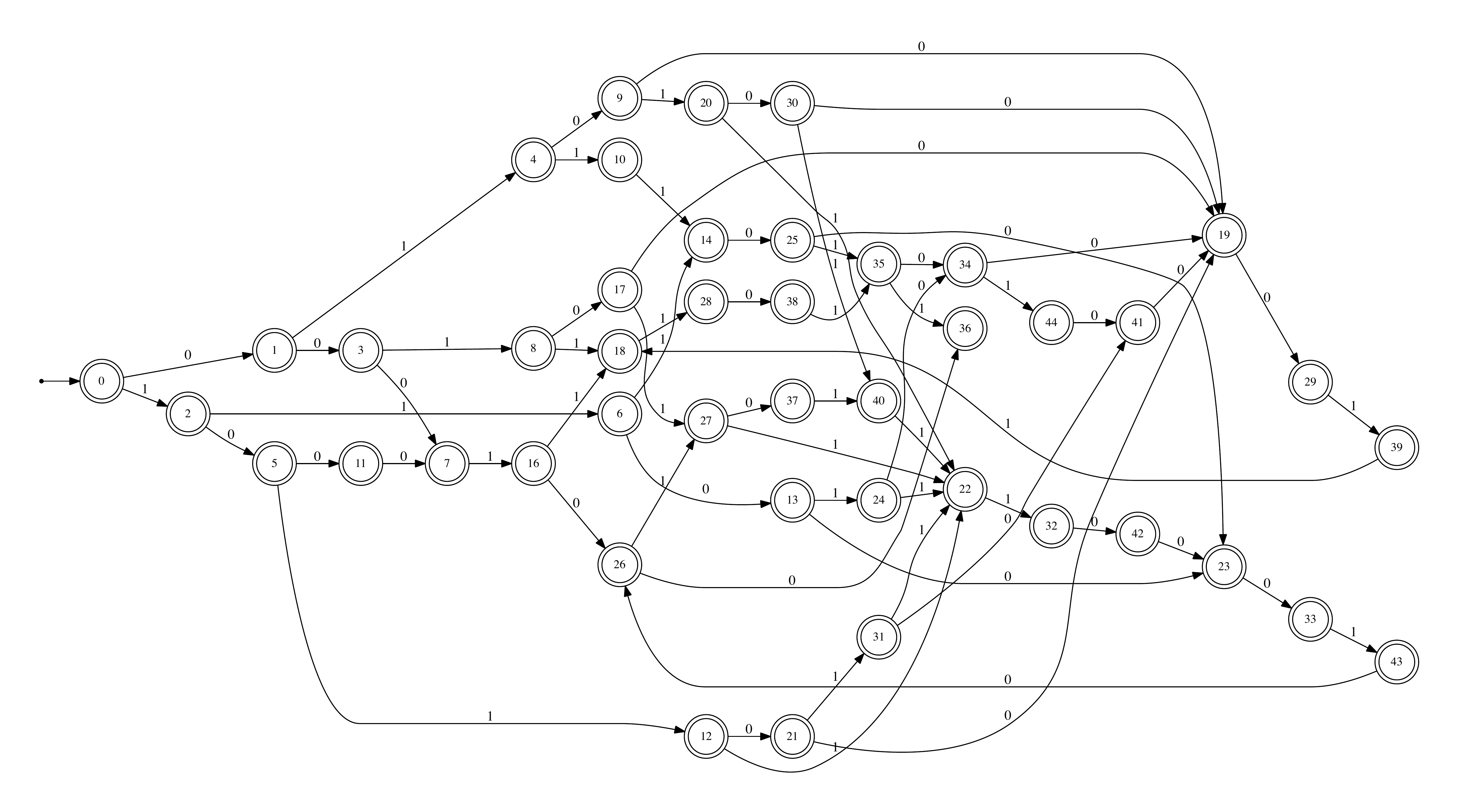}
\end{center}
\caption{Automaton for binary words with 
longest even palindrome factor of length $\leq 2$
and longest odd palindrome of length $\leq 5$.}
\label{palm2-25}
\end{figure}

\begin{proof}
We construct the automaton for $R_{2,5}(\Sigma_2)$ as discussed above.
Before minimization it has 155 states.
After minimization it has 44 states. 
State 18 is birecurrent, with cycles
labeled $x_0 = 10100011$ and $x_1 = 1010100011$.
\end{proof}

\begin{theorem}
Let $(r_{2,2,5}(n))_{n \geq 0}$ denote the number of finite binary
words containing 
longest even palindrome factor of length $\leq 2$
and longest odd palindrome of length $\leq 5$.
Then
$r_{2,2,5}(n) = r_{2,2,5}({n-8}) + r_{2,2,5}({n-10})$ for $n \geq 16$.  
Furthermore,
$r_{2,2,5}(n) \sim C_1 \alpha^n + C_2 (-\alpha)^n$,
$C_1 \doteq 15.991809$, $C_2 \doteq 0.023895$,
and $\alpha \doteq 1.0804184273981 $ is the largest
real zero of $X^{10} - X^2 - 1$.
\end{theorem}

\begin{proof}
The minimal polynomial of the corresponding matrix is
$$X^6(X-2)(X^{10} - X^2 - 1).$$
The minimal annihilator of the recurrence can be
determined by using the ideas in Section~\ref{three}; it is
$X^{10} - X^2 - 1$.
\end{proof}

\begin{remark}
The sequence $r_{2,2,5} (n)$ is sequence
\seqnum{A330130} in the OEIS.
\end{remark}

The case of longest even palindrome factor of length $\leq 4$
and longest odd palindrome of length $\leq 3$ is already
covered in Theorem~\ref{pal2len4-thm}.

\begin{theorem}
There are uncountably many uniformly recurrent binary words over
having
longest even palindrome factor of length $\leq 6$
and longest odd palindrome of length $\leq 3$.
\end{theorem}

\begin{figure}[H]
\begin{center}
\includegraphics[width=6.5in]{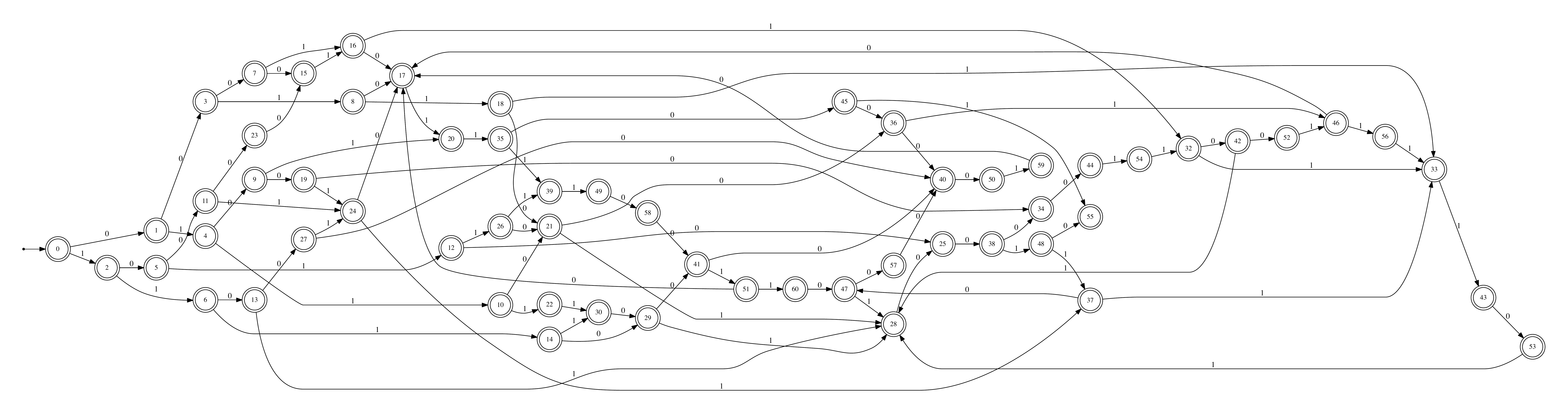}
\end{center}
\caption{Automaton for binary words with 
longest even palindrome factor of length $\leq 6$
and longest odd palindrome of length $\leq 3$.}
\label{palm2-63}
\end{figure}

\begin{proof}
We construct the automaton for $R_{6,3}(\Sigma_2)$ as discussed above.
Before minimization it has 477 states.
After minimization it has 60 states.
State 17 is birecurrent, with cycles
labeled $x_0 = 110010$ and $x_1 = 1111000010$.
\end{proof}

\begin{theorem}
Let $(r_{2,6,3}(n))_{n \geq 0}$ denote the number of finite binary
words containing 
longest even palindrome factor of length $\leq 6$
and longest odd palindrome of length $\leq 3$.
Then
$r_{2,6,3}(n) = r_{2,6,3}({n-6}) + 2r_{2,6,3}({n-8}) + 3r_{2,6,3}({n-10})
+ r_{2,6,3}({n-14})$ for $n \geq 21$.  Furthermore,
and $r_{2,6,3}(n) \sim C_1 \alpha^{n} + C_2 (-\alpha)^n$,
where
$C_1 \doteq 11.58110542$, 
$C_2 \doteq 0.00264754$, 
and $\alpha \doteq 1.244528319539183$ is the largest
real zero of $X^{14} - X^8 -2X^6 - 3X^4 - 1$.
\end{theorem}

\begin{proof}
The minimal polynomial of the corresponding matrix is
$$X^7 (X-2)(X^2 + 1)(X^{14} - X^8 -2X^6 - 3X^4 - 1)(X^{12} - X^{10} + X^8 - 2X^6 + X^2 - 1).$$
The minimal annihilator of the recurrence can be
determined by using the ideas in Section~\ref{three}; it is
$X^{14} - X^8 -2X^6 - 3X^4 - 1$.
\end{proof}

\begin{remark}
The sequence $r_{2,6,3}$ is sequence \seqnum{A330131} in the OEIS.
\end{remark}

\subsection{Alphabet size 3}

\begin{theorem}
There are uncountably many uniformly recurrent words over
$\Sigma_3$ containing no (nonempty) even palindromic factors 
and longest odd palindrome of length $\leq 3$.
\end{theorem}

\begin{figure}[H]
\begin{center}
\includegraphics[width=6.5in]{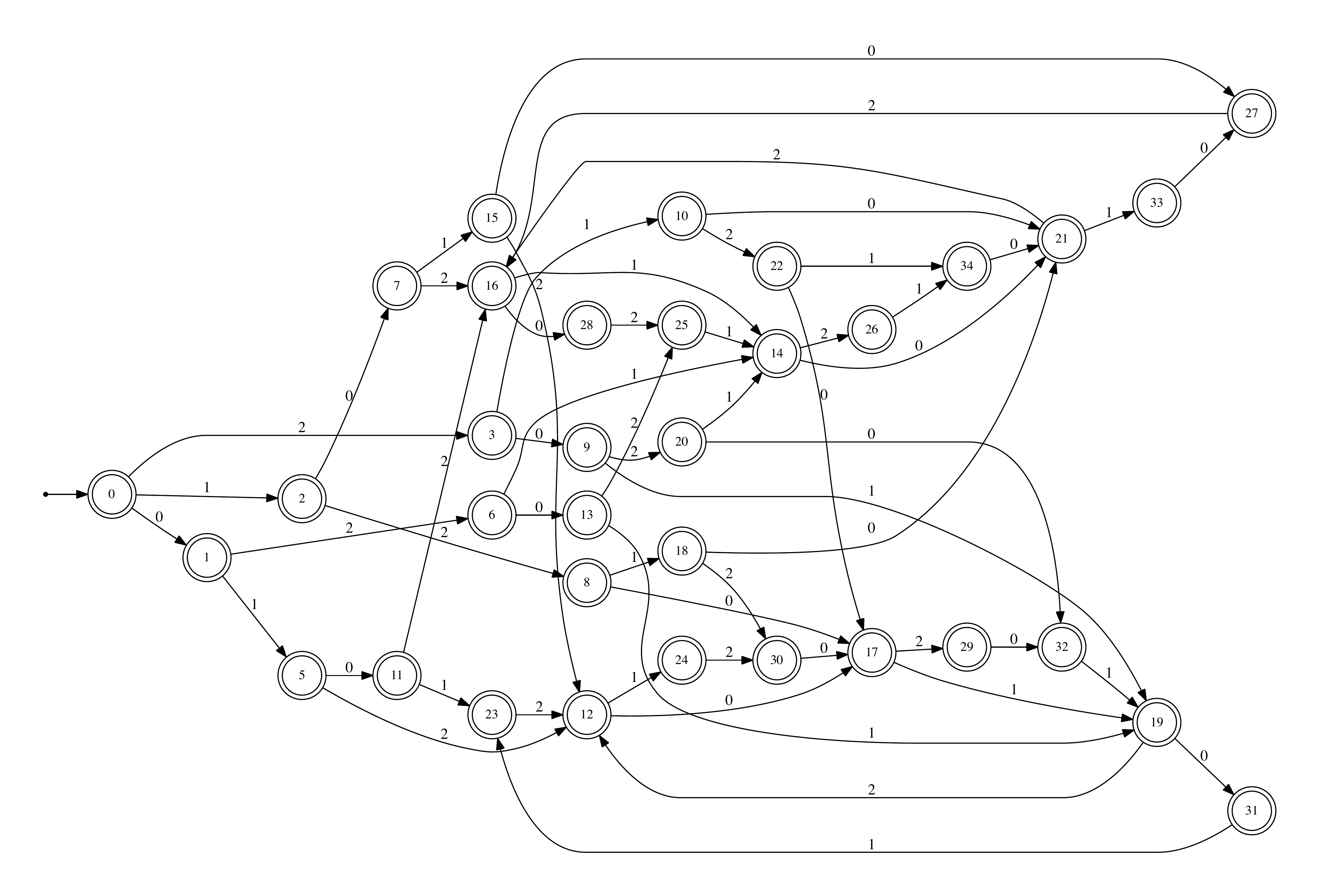}
\end{center}
\caption{Automaton for ternary words with no even palindromic
factors and longest odd palindrome of length $3$}
\label{palm3-03}
\end{figure}

\begin{proof}
We construct the automaton for $R_{0,3} (\Sigma_3)$ as discussed 
in Section~\ref{two}.
Before minimization it has 88 states.
After minimization it has 34 states.
State 16 is birecurrent, with cycles
labeled $x_0 = 021210102$ and $x_1 = 1210102$.
\end{proof}

\begin{theorem}
Let $(r_{3,0,3} (n))_{n \geq 0}$ denote the number of finite ternary
words containing no (nonempty) even palindromic factors 
and longest odd palindrome of length $3$.  Then
$$r_{0,3} (n) = r_{0,3}({n-1}) + r_{0,3} ({n-3})$$
for $n \geq 7$.  Furthermore, $r_{0,3} (n) \sim C \alpha^n$, where
$C \doteq 5.37711043$ 
and $\alpha \doteq 1.465571231876768$ is the largest
real zero of $X^3 - X^2 - 1$.
\end{theorem}

\begin{proof}
The minimal polynomial of the corresponding matrix is
$$X^4 (X-3) (X^2 - X + 1)(X^3 - X^2 - 1)(X^4 + 2X^3 + 2X^2 + X + 1).$$
The minimal annihilator of the recurrence can be
determined by using the technique in Section~\ref{three}; it is
$X^3 - X^2 - 1$.
\end{proof}

\begin{remark}
The sequence is  \seqnum{A330132} in the OEIS.
$r_{0,3} (n) = 6 \cdot$ \seqnum{A000930}$(n-1)$ for $n \geq 5$, where
\seqnum{A000930} is the Narayana cow sequence.
\end{remark}

The case of largest even palindrome of length $2$ and largest
odd palindrome of length $1$ is already covered in
Theorem~\ref{pal3len2-thm}.

\section{Number of odd and even palindromes}

Our final application is to infinite words containing a specified
number of even and odd palindromes.   We define
$t_{k,\ell,m}(n)$ to be the number of length-$n$ words
in $T_{\ell,m} (\Sigma_k)$.

\subsection{Alphabet size 2}

Here, instead of providing the
details, we simply summarize our results in tabular form.  The
minimal annihilators for the
sequences can be computed from the data we computed.

The following cases have infinite words, but not aperiodic infinite
words.

\begin{table}[H]
\centering
\begin{tabular}{|c|c|c|c|c|c|c|}
Max number of & Max number of & States & States & Example word  \\
even palindromes & odd palindromes & (unminimized) & (minimized) & \\
\hline
3 & 9 & 10795 & 1468 & $01(00010111)^\omega$ \\
3 & 8 & 3911 & 799 & $1(00010111)^\omega$ \\
4 & 7 & 7505 & 1181 & $01(0001011)^\omega$ \\
4 & 6 & 2413 & 530 & $1(0001011)^\omega$ \\
5 & 5 & 1647 & 419 & $0 (001011)^\omega$ \\
5 & 4 & 461 & 136 & $(001011)^\omega$ \\
6 & 5 & 3141 & 604 & $(00001011)^\omega$ \\
6 & 4 & 699 & 177 & $0(011001)^\omega$  \\
7 & 4 & 1081 & 261 & $10(011001)^\omega$ \\
8 & 4 & 1729 & 375 & $1101(001011)^\omega$ \\
\end{tabular}
\end{table}

The following cases have examples of aperiodic infinite words.

\begin{table}[H]
\centering
\resizebox{\columnwidth}{!}{%
\begin{tabular}{|c|c|c|c|c|c|c|}
Max number of & Max number of & States & States & $x_0$ & $x_1$ & Birecurrent \\
even palindromes & odd palindromes & (unminimized) & (minimized) & & & state number \\
\hline
3 & 10 & 33685 & 3071 & 00011101 & 0100011101 & 1836 \\
4 & 8 & 26937 & 2830 & 0010111 & 00010111 & 2364 \\
5 & 6 & 7495 & 1269 & 001011 & 0001011 & 1035  \\
7 & 5 & 6741 & 955 & 001011 & 00001011 & 904 \\
9 & 4 & 2789 & 545 & 001011 & 0011001011 & 450 \\
\end{tabular}%
}
\end{table}

\subsection{Alphabet size 3}

The only interesting case is one even palindrome and
five odd palindromes.  Here the automaton has 6208 states
(632 when minimized) and has a birecurrent state,
corresponding to $x_0 = 01012$ and $x_1 = 012$.

\section{Conclusions}

We have reproved most of the theorems in \cite{Fici&Zamboni:2013} using
a unified approach based on finite automata.  This is evidence for
the thesis, previously announced in \cite{Rajasekaran&Shallit&Smith:2019},
that long case-based arguments are good candidates for replacement by 
algorithms and logical decision procedures.

All of the code referred to in this paper is available at \\
\centerline{\url{https://cs.uwaterloo.ca/~shallit/papers.html} \ .}

\newcommand{\noopsort}[1]{} \newcommand{\singleletter}[1]{#1}

\end{document}